\documentclass{amsproc}
\usepackage{amsmath}
\usepackage{amsfonts}
\usepackage{amssymb}
\usepackage{amsthm}
\usepackage{newlfont}
\usepackage{graphicx}
\usepackage{amscd}

\theoremstyle{plain}
\newtheorem{Th}{Theorem}[section]
\newtheorem{Cor}[Th]{Corollary}
\newtheorem{Lem}[Th]{Lemma}
\newtheorem{Prop}[Th]{Proposition}
\theoremstyle{definition}
\newtheorem{Def}{Definition}[section]

\theoremstyle{remark}
\newtheorem*{Rem}{Remark}
\numberwithin{equation}{section}
\newcommand{\PP}{{\mathbb P}}

\newcommand{\DD}{{\mathbb D}}

\newcommand{\ZZ}{{\mathbb Z}}
\newcommand{\VV}{{\mathbb V}}
\newcommand{\RR}{{\mathbb R}}

\newcommand{\XX}{{\mathbb X}}

\newcommand{\cL}{{\mathcal L}}

\newcommand{\bphi}{\boldsymbol{\phi}}
\newcommand{\bPhi}{\boldsymbol{\Phi}}

\begin{document}

\title{Hirota equation and the quantum plane}

\author{Adam Doliwa}
\address{Faculty of Mathematics and Computer Science, University of Warmia and Mazury, ul.~S{\l}oneczna~54, 10-710 Olsztyn, Poland}
\email{doliwa@matman.uwm.edu.pl}
\urladdr{http://wmii.uwm.edu.pl/~doliwa/}

\date{\today}
\keywords{integrable discrete geometry; Hirota equation; Desargues maps; Darboux transformations; affine Weyl group action; multidimensional quadrilateral lattice; pentagonal relation; Weyl commutation relations; quantum plane}
\subjclass[2010]{Primary 37K10; Secondary 39A14, 37K60, 51A20, 16T20}

\begin{abstract}
We discuss geometric integrability of Hirota's discrete KP equation in the framework of projective geometry over division rings using the recently introduced notion of Desargues maps. We also present the Darboux-type transformations, and we review symmetries of the Desargues maps from the point of view of root lattices of type $A$ and the action of the corresponding affine Weyl group. Such a point of view facilities to study the relation of Desargues maps and the discrete conjugate nets. Recent investigation of geometric integrability of Desargues maps allowed to introduce two maps satisfying functional pentagon equation. Moreover, the ultra-locality requirement imposed on the maps leads to Weyl commutation relations. We show that the pentagonal property of the maps allows to define a coproduct in the quantum plane bi-algebra, which can be extended to the corresponding Hopf algebra.
\end{abstract}
\maketitle
\begin{quote}
\emph{The relevance of a geometric theorem is determined by what the theo\-rem tells us about space, and not by the eventual difficulty of the proof. The Desargues' theorem of projective geometry comes as close as a proof can to the Zen ideal. It can be summarized in two words: "I~see!" Nevertheless, Desargues' theorem, far from trivial despite the simplicity of its proof, has many more applications both in geometry and beyond ...}

\bigskip \hfill \emph{Gian Carlo Rota, The Phenomenology of Mathematical Proof}
\end{quote}

\section{Introduction}
Everybody interested in the (pre)history of soliton theory should consult monographs of Bianchi~\cite{Bianchi}, Darboux~\cite{Darboux-OS,DarbouxIV}, Eisenhart~\cite{Eisenhart-TS} or Tzitz\'{e}ica~\cite{Tzitzeica}. In these geometry books, which summarize classical XIX-th century style developments in theory of submanifolds and their transformations, one can recognize many fundamental facts from the theory of integrable partial differential equations. In looking for analogous geometric interpretation of integrable partial difference systems we have found that very often their integrability features are encoded in incidence geometry theorems of Pappus, Desargues, Pascal, Miquel and others \cite{CDS,Dol-Koe,BQL,CQL,Dol-Des}, compare also works \cite{BobSur-Proc,BobSur,BobSur-Koenigs,UHJ,KingSchief1,KingSchief,KoSchief-Men,KoSchiefSBKP,KoSchiefSDS-II} written in a similar spirit; for introduction to projective geometry and its subgeometries see \cite{Coxeter-IG,Richter-Gebert}. 

Hirota's discrete Kadomtsev--Petviashvili (KP) equation \cite{Hirota} may be considered as the Holy Grail of integrable systems theory, both on the classical and the quantum level \cite{KNS-rev}. In the present paper, based on our earlier publications \cite{Dol-Des,Dol-AN,DoliwaSergeev-pentagon}, we review geometric aspects of the non-commutative Hirota system within the framework of projective geometry over division rings. The crucial notion here is that of Desargues maps, where the underlying geometric property is collinearity of three points, which gives the linear problem for the Hirota system. This should be considered as further simplification of (already rather non-complicated) approach to integrable discrete geometry via the theory of multidimensional quadrilateral lattices~\cite{MQL} based on coplanarity of four points. We remark that the quadrilateral lattice is the integrable discrete analogue~\cite{Sauer2,DCN} of the conjugate net which is the fundamental geometric object of the geometric works of Darboux and his contemporaries mentioned above.
Surprisingly enough, the theory of quadrilateral lattices is contained in the theory of Desargues maps. Moreover, the latter allows for a description in terms of the $A$-type root lattice what makes it invariant with respect to the corresponding affine Weyl group action.    

One of motivations to study non-commutative versions of integrable discrete systems \cite{FWN-Capel,Kupershmidt,BobSur-nc,Nimmo-NCKP} is their relevance in integrable lattice field theories. In particular, as shown by \cite{BaMaSe,Sergeev-q3w}, the four dimensional consistency of the geometric construction of quadrilateral lattice \cite{MQL} is related to Zamolodchikov's tetrahedron equation~\cite{Zamolodchikov}, which is a multidimensional analogue of the quantum Yang--Baxter equation~\cite{Baxter,QISM,YB-Miwa}. A closer look at the structure of Desargues maps allowed to isolate~\cite{DoliwaSergeev-pentagon} two maps which satisfy the functional pentagon equation, whose quantum version forms the fundamental ingredient in the present-day theory of quantum groups~\cite{BaajSkandalis,Woronowicz-P,KustermansVaes}. As it was shown in~\cite{DoliwaSergeev-pentagon}, in the transition "from non-commutative to quantum" it is enough to assume the ultra-locality~\cite{PCT} of the maps, i.e. in the context of Desargues maps \emph{ultra-locality requires the Weyl commutation relations~\cite{Weyl-GQ}}.  

To previously published results we added here the geometric meaning of the Darboux-type transformations of the non-commutative Hirota equation. Moreover we study in more detail the connection of the pentagon maps related to the Hirota equation to the bialgebra structure of Manin's  quantum plane \cite{Manin} and the Hopf algebra structure of quantum '$ax+b$' group~\cite{ax+b}. 

\section{Integrability and symmetry of the Hirota equation}
\subsection{The Hirota equation and Desargues maps}
Consider \cite{Dol-Des} Desargues maps 
$\phi\colon\ZZ^N\to\PP^M(\DD)$ of multidimensional integer lattice into 
projective space of dimension $M\geq 2$ over division ring $\DD$, characterized by the condition that for an arbitrary $n\in\ZZ^N$ and any pair of indices 
$i\ne j$ the points $\phi(n)$, $\phi(n+\boldsymbol{\varepsilon}_i)$ and 
$\phi(n+\boldsymbol{\varepsilon}_j)$ are collinear; here 
$\boldsymbol{\varepsilon}_i = (0, \dots, \stackrel{i}{1}, \dots ,0)$
is the $i$-th element of the canonical basis of $\RR^N$. In what follows we will use the standard notation $F_{(\pm i)}(n)= F(n\pm\boldsymbol{\varepsilon}_i)$ for any function $F$ on $\ZZ^N$. We will also often skip the argument $n$ of the map.
\begin{figure}
\begin{center}
\includegraphics[width=7cm]{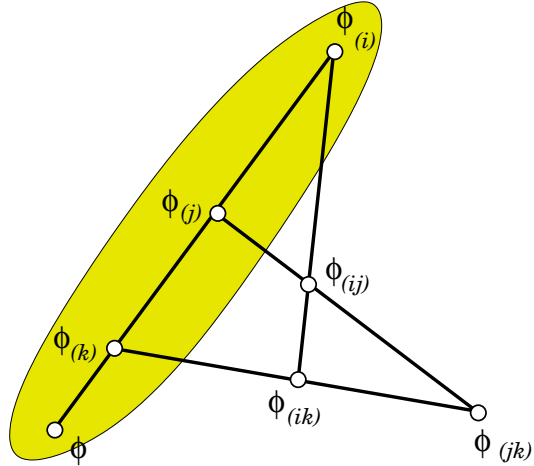}
\end{center}
\caption{Desargues map condition with the Veblen configuration. The points $\phi$, $\phi_{(i)}$, $\phi_{(j)}$ and $\phi_{(k)}$ correspond to the points 
$A$, $B$, $C$ and $D$ used to define the normalization map~$W$} 
\label{fig:Desargues-Veblen-normal}
\end{figure}
In the homogeneous coordinates $\bphi\colon\ZZ^N\to\DD^{M+1}$ the defining condition of the Desargues maps can be described in terms of the linear system
\begin{equation} \label{eq:lin-A}
\bPhi + \bPhi_{(i)} A_{ij} + \bPhi_{(j)} A_{ji} = 0, \qquad i\ne j,
\end{equation}
where $A_{ij}\colon\ZZ^N\to {\DD}^{\times}$ are certain non-vanishing functions (we consider the generic case where no two of the three collinear points coincide).
The compatibility of the linear system \eqref{eq:lin-A}, or the Desargues map equation, reads
\begin{align} \label{eq:alg-cond}
& A_{ij}^{-1} A_{ik} + A_{kj}^{-1} A_{ki} = 1, \\
\label{eq:shift-cond}
&A_{ik(j)}A_{jk} = A_{jk(i)} A_{ik},
\end{align}
where indices $i,j,k$ are distinct. 

As it was shown in \cite{Dol-Des} there exists a special gauge in which
the linear problem takes the form
\cite{DJM-II,Nimmo-NCKP}
\begin{equation} \label{eq:lin-dKP}
\bPhi_{(i)} - \bPhi_{(j)} =  \bPhi U_{ij},  \qquad i \ne j \leq N.
\end{equation}
Then equations
\eqref{eq:alg-cond}-\eqref{eq:shift-cond} reduce
to the following systems \cite{Nimmo-NCKP} for distinct triples $i,j,k$
\begin{align} \label{eq:alg-comp-U}
& U_{ij} + U_{jk} + U_{ki} = 0, \\
& \label{eq:U-rho} 
U_{kj}U_{ki(j)} = U_{ki} U_{kj(i)}.
\end{align}
Equation \eqref{eq:U-rho} allows to introduce 
the potentials $\rho_i\colon\ZZ^K\to\DD^\times$ such that
\begin{equation} \label{eq:def-rho}
U_{ij} = \rho_i^{-1}
 \rho_{i(j)}.
\end{equation}
When $\DD$ is commutative, i.e. a field, the functions $\rho_i$ can 
be parametrized in terms of a 
single potential $\tau$ (the tau-function)
\begin{equation} \label{eq:r-tau}
\rho_i = (-1)^{\sum_{k>i}n_k}
\frac{\tau_{(i)}}{\tau},
\end{equation} 
and then the functions $U_{ij}$ can 
be parametrized as follows
\begin{equation} \label{eq:U-tau}
U_{ij} = \frac{\tau \tau_{(ij)}}{\tau_{(i)} \tau_{(j)}}, \qquad i< j, 
\end{equation}
which solves equation \eqref{eq:U-rho}. Then equations
\eqref{eq:alg-comp-U} reduce to the celebrated Hirota system \cite{Hirota}
\begin{equation} \label{eq:H-M}
\tau_{(i)}\tau_{(jk)} - \tau_{(j)}\tau_{(ik)} + \tau_{(k)}\tau_{(ij)} =0,
\qquad 1\leq i< j <k \leq N.
\end{equation}
\begin{Rem}
The linear system \eqref{eq:lin-A} written~\cite{Dol-Des} in terms of nonhomogeneous coordinates (the affine gauge) leads to the non-commutative discrete modified KP system  \cite{FWN-Capel}.
\end{Rem}

A crucial property of the Hirota equation (originally written for $N=3$) is that the number of independent variables can be arbitrary large. Such a multidimensional consistency is nowadays placed at the central point \cite{ABS,FWN-cons} of integrability theory and is considered as "the precise analogue of the \emph{hierarchy} of nonlinear evolution equations in the case of continuous systems"~\cite{NRGO}. We remark that already in the first works on geometry of quadrilateral lattices, their reductions and transformations~\cite{MQL,CDS,TQL,q-red,DS-sym} the multidimensional aspects of the procedure of construction (called there the geometric integrability scheme) of the lattice from initial data were considered, and the identification of the Darboux-type transformations as recursive augmentation of the number of independent variables has been pointed out.

The important relation of four dimensional consistency of the Hirota equation (in its Schwarzian form) and the Desargues configuration has been observed by Wolfgang Schief (the author has learned about that relation from the talk of A.~Bobenko~\cite{Bobenko-talk}).

\subsection{The Darboux transformations of the Hirota equation}
Below we present the theory of Darboux-type transformations of Desargues maps and of the corresponding solutions of the non-commutative Hirota equation (see also~\cite{GNO}). Its geometric content follows closely the theory of transformations of quadrilateral lattices~\cite{MDS,TQL}, for analogous algebraic results in the commutative context see \cite{Nimmo-KP,Doliwa-Nieszporski}.
\subsubsection{Elementary Darboux transformation of the Desargues maps}
\begin{Prop}
Given a solution $\theta\colon\ZZ^N\to\DD^\times$ of the linear system \eqref{eq:lin-dKP} of the non-commutative Hirota equation then
\begin{equation} \label{eq:DT}
\bPhi^{\mathcal{D}} = \bPhi_{(i)} - \bPhi \, \theta^{-1} \theta_{(i)}
\end{equation}
satisfies the linear system of the same form
\begin{equation} \label{eq:lin-dKP-D}
\bPhi_{(i)}^{\mathcal{D}}  - \bPhi_{(j)}^{\mathcal{D}}  =  
\bPhi^{\mathcal{D}}  U_{ij}^{\mathcal{D}} ,  \qquad i \ne j \leq N,
\end{equation}
with the transformed potential 
\begin{equation} \label{eq:U-D}
U_{ij}^{\mathcal{D}} = \theta^{-1}_{(i)} \, \theta U_{ij}\,  \theta^{-1}_{(j)} \theta_{(ij)} , \qquad i\neq j.
\end{equation}
\end{Prop}
\begin{Cor} \label{cor:DT}
Let $\bPhi^{\mathcal{D}} $ be constructed from  $\bPhi$ as above.\\
1) The transform $\bPhi^{\mathcal{D}} $ of $\bPhi$ does not depend on the index $i$.\\
2) Points of the Desargues map described by $\bPhi^{\mathcal{D}}$, the corresponding points of the map represented by $\bPhi_{(i)}$ and the points representing $\bPhi$ are collinear for all $1\leq i\leq N$.\\
3) The potentials $\rho_i$ transform according to the formula
\begin{equation} \label{eq:rho-D}
\rho_i^{\mathcal{D}}  = \rho_i \theta^{-1} \theta_{(i)}.
\end{equation}
4) The function $\theta^{-1}$ satisfies equations
\begin{equation}
\theta^{-1}_{(j)} - \theta^{-1}_{(i)} = U_{ij}^{\mathcal{D}} \theta^{-1}_{(ij)} ,
\qquad i\neq j
\end{equation}
\end{Cor}
\begin{Def}
The Desargues map $\phi^{\mathcal{D}}\colon \ZZ^N \to \PP^M(\DD)$ is called \emph{elementary Darboux transform} of the Desargues map $\phi \colon \ZZ^N \to \PP^M(\DD)$ if points of $\phi^{\mathcal{D}}$ are incident with the lines of nearest neighbours of the corresponding points of $\phi$.
\end{Def}
\begin{Rem}
Geometric meaning of the formula \eqref{eq:DT} is as follows~\cite{TQL}. Using $\theta$ we extend the Desargues map represented by $\bPhi$ from $\PP^M(\DD)$ to the map represented by $\left( \begin{array}{c} \theta \\ \bPhi \end{array} \right)$ in the space of one dimension more. Then its elementary Darboux transform is given by intersection of the line of positive neighbours of the extended map with its natural projection (the line of positive neighbours of $\bPhi$) on the initial hyperspace
\begin{equation}
\left( \begin{array}{c} \theta_{(i)} \\ \bPhi_{(i)} \end{array} \right) + 
\left( \begin{array}{c} \theta \\ \bPhi \end{array} \right) \lambda = \left( \begin{array}{c} 0 \\ \bPhi^{\mathcal{D}} \end{array} \right).
\end{equation}
From that it follows that in the generic case elementary Darboux transforms can be described in algebraic terms as above.
\end{Rem}

\subsubsection{Elementary adjoint Darboux transformation of Desargues maps} 
\begin{Def}
The Desargues map $\phi^{\mathcal{D}^*} \colon \ZZ^N \to \PP^M(\DD)$ is called \emph{elementary adjoint Darboux transformation} of the Desargues map $\phi\colon \ZZ^N \to \PP^M(\DD)$ if $\phi$ is Darboux transformation of points of $\phi^{\mathcal{D}^*}$.
\end{Def}
\begin{figure}
\begin{center}
\includegraphics[width=16cm]{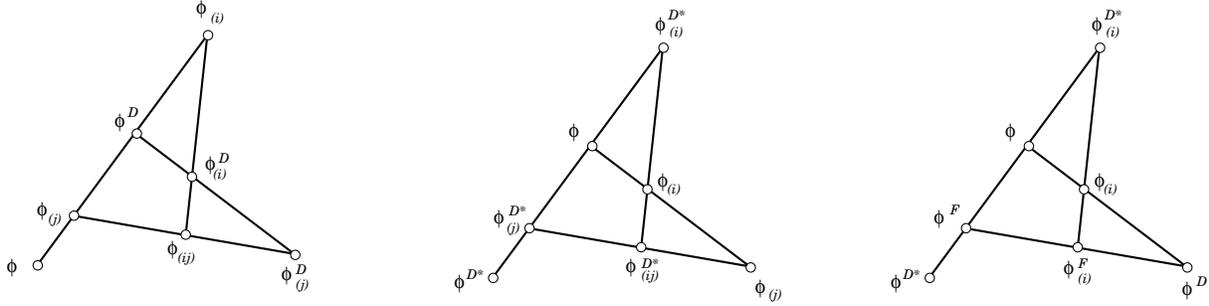}
\end{center}
\caption{The Darboux, adjoint Darboux, and binary Darboux (fundamental) transformations of the Desargues maps} 
\label{fig:Desargues-Veblen-Darboux}
\end{figure}
We will derive the algebraic formula of the transformation. Let $\tilde{\theta}$ be the solution of the linear system satisfied by 
$\tilde{\bPhi}=\bPhi^{\mathcal{D}^*} $ (in the Hirota gauge) which defines the corresponding solution $\tilde{\bPhi}^\mathcal{D}=\bPhi$ of its Darboux transform, then by point 4) of Corollary~\ref{cor:DT} above the function $\theta^*=\tilde{\theta}^{-1}$ satisfies the equation
\begin{equation} \label{eq:lin-dKP-a}
\theta^*_{(j)} - \theta^*_{(i)} = U_{ij} \theta^*_{(ij)} , \qquad i\neq j,
\end{equation}
called the adjoint of equation \eqref{eq:lin-dKP}. To proceed further we state the Lemma, which can be demonstrated by direct calculation \cite{MDS}.  
\begin{Lem}
Given (column-vector) solution $\boldsymbol{\Phi}:\mathbb{Z}^N \to\DD^M$, 
of the linear system \eqref{eq:lin-dKP}, and given (row-vector) solution 
$\boldsymbol{\Phi}^*:\mathbb{Z}^N \to (\DD^K)^T$, of the adjoint linear system 
\eqref{eq:lin-dKP-a}. These allow to
construct the matrix valued potential 
$\Omega[\boldsymbol{\Phi},\boldsymbol{\Phi}^*]:
\mathbb{Z}^N\to \mathrm{Mat}^M_K (\DD)$,
defined by the compatible system
\begin{equation} \label{eq:omega-p-p}
\Omega[\boldsymbol{\Phi},\boldsymbol{\Phi}^*]_{(i)} -  
\Omega[\boldsymbol{\Phi},\boldsymbol{\Phi}^*]= 
\boldsymbol{\Phi} \boldsymbol{\Phi}^*_{(i)}, 
\qquad i = 1,\dots , N.
\end{equation} 
\end{Lem}
\begin{Prop}
Given solution $\theta^*\colon\ZZ^N\to \DD^\times$ of the adjoint equation \eqref{eq:lin-dKP-a} of the linear system \eqref{eq:lin-dKP} of the Desargues map $\phi\colon \ZZ^N \to \PP^M(\DD)$ then 
\begin{equation} \label{eq:DT-a}
\bPhi^{\mathcal{D}^*} = \Omega[\bPhi,\theta^*] \theta^{*-1},
\end{equation}
represents the elementary Darboux transform $\phi^{\mathcal{D}^*} \colon \ZZ^N \to \PP^M(\DD)$, and satisfies the linear system \eqref{eq:lin-dKP} with the potentials
\begin{equation} \label{eq:U-D*}
U_{ij}^{\mathcal{D}^*} = \theta^* \theta^{*-1}_{(i)} \, U_{ij}\,  \theta^{*}_{(ij)} \theta^{*-1}_{(j)}  , \qquad 
\rho_{i}^{\mathcal{D}^*} = \rho_i \theta^*_{(i)} \theta^{*-1}.
\end{equation}
\end{Prop}
\begin{proof}
It is enough to check that formula \eqref{eq:DT} in the present notation reads
\begin{equation*}
\bPhi = \left( \left( \bPhi^{\mathcal{D}^*} \theta^* \right)_{(i)} - 
\bPhi^{\mathcal{D}^*} \theta^* \right) \theta^{*-1}_{(i)}.
\end{equation*}
Transformation rules of the potentials can be derived directly, but actually they follow from interpretation of equations \eqref{eq:U-D} and \eqref{eq:rho-D} in the present notation.
\end{proof}
The duality between the elementary Darboux transformation and the adjoint elementary Darboux transformation follows from the transformation formulas below, which will be useful in the next section.
\begin{Prop}
Given a solution $\bPhi^*\colon\ZZ^N\to(\DD^K)^T$ of the adjoint linear system \eqref{eq:lin-dKP-a} and given solutions $\theta\colon\ZZ^N\to\DD^\times$ and $\theta^*\colon\ZZ^N\to\DD^\times$ of equations \eqref{eq:lin-dKP} and \eqref{eq:lin-dKP-a}, correspondingly. Then\\
1) the function
\begin{equation*}
\bPhi^{*{\mathcal{D}}} = \theta^{-1}\Omega[\theta,\bPhi^*]
\end{equation*}
satisfies the adjoint linear system \eqref{eq:lin-dKP-a} with the potentials $U^\mathcal{D}_{ij}$ given by \eqref{eq:U-D};\\
2)  the function
\begin{equation*}
\bPhi^{*{\mathcal{D}^*}} = \bPhi^*_{(-i)} - \theta^{*}_{(-i)} \theta^{*\, -1} \bPhi^*
\end{equation*}
does not depend on the index $i$ and satisfies the adjoint linear system \eqref{eq:lin-dKP-a} with the potentials $U^\mathcal{D^*}_{ij}$ given by \eqref{eq:U-D*}.
\end{Prop}
\subsubsection{The vectorial binary Darboux transformations and their superpositions}
\begin{Def}
A superposition of elementary Darboux transformation and adjoint elementary Darboux transformation of Desargues map (or equivalently of solution of the Hirota system) is called \emph{binary Darboux (or fundamental) transformation} of Desargues map (of the Hirota system).
\end{Def}
We remark that such transformations play fundamental role \cite{TQL} in geometric theory of transformations of multidimensional lattices of planar quadrilaterals (or discrete Darboux system), see also section~\ref{sec:D-QL}. Before doing any calculation let us demonstrate that the quadrilateral whose vertices are the corresponding points $\phi$, $\phi_{(i)}$ and its binary Darboux transformed $\phi^{\mathcal{F}}$, $\phi_{(i)}^{\mathcal{F}}$ are coplanar. Indeed, both $\phi$ and $\phi^{\mathcal{F}}$ are elementary Darboux transformations of $\phi^{\mathcal{D}^*}$, and simultaneously elementary adjoint Darboux transformations of $\phi^{\mathcal{D}}$ as visualised on Figure~\ref{fig:Desargues-Veblen-Darboux}. 

Let us present the algebraic description of the binary Darboux transformation, which can be checked by direct calculation, see also \cite{MDS,TQL,Doliwa-Nieszporski}.
\begin{Prop}
Given solution $\theta\colon\ZZ^N\to \DD^\times$ of the linear system \eqref{eq:lin-dKP} of the Desargues map $\phi : \ZZ^N \to \PP^M(\DD)$ represented by solution $\bPhi\colon \ZZ^N \to \DD^{M+1}$ of \eqref{eq:lin-dKP}, and given solution 
$\theta^*\colon\ZZ^N\to \DD^\times$ of its adjoint\eqref{eq:lin-dKP-a} then\\
1) the vector valued function
\begin{equation*} 
\bPhi^{\mathcal{F}} = \bPhi - \Omega[\bPhi,\theta^*] \Omega[\theta,\theta^*]^{-1} \theta,
\end{equation*}
is a solution of the linear system \eqref{eq:lin-dKP} with the transformed potentials
\begin{equation*}
U_{ij}^{\mathcal{F}} = U_{ij} - \left( \theta^* \Omega[\theta,\theta^*]^{-1} \theta  \right)_{(i)} + \left( \theta^* \Omega[\theta,\theta^*]^{-1} \theta  \right)_{(j)} ,
\qquad \rho_i^{\mathcal{F}} = \rho_i \left( 1 +  
\theta^*_{(i)} \Omega[\theta,\theta^*]^{-1} \theta \right);
\end{equation*}
2) it is obtained as superposition of transformations
\begin{equation*}
\bPhi \stackrel{\theta}{\longrightarrow} \bPhi^{\mathcal{D}} 
\stackrel{\theta^{*\mathcal{D}}}{\longrightarrow} (\bPhi^{\mathcal{D}})^{\mathcal{D^*}} = \bPhi^{\mathcal{F}}, \qquad
\theta^{*\mathcal{D}} = \theta^{-1} \Omega[\theta,\theta^*];
\end{equation*}
3) it is obtained as superposition of transformations
\begin{equation*}
\bPhi \stackrel{\theta^*}{\longrightarrow} \bPhi^{\mathcal{D}^*} 
\stackrel{\theta^{\mathcal{D}^*}}{\longrightarrow} (\bPhi^{\mathcal{D}^*})^{\mathcal{D}} = \bPhi^{\mathcal{F}}, \qquad
\theta^{\mathcal{D}^*} = \Omega[\theta,\theta^*] \theta^{*-1} ;
\end{equation*}
3) the corresponding fundamental transformation of the solution 
$\bPhi^*:\ZZ^N\to(\DD^K)^T$ of the adjoint linear system \eqref{eq:lin-dKP-a} reads
\begin{equation*} 
\bPhi^{*\mathcal{F}} = \bPhi^* - \theta^* \Omega[\theta,\theta^*]^{-1} \Omega[\theta,\bPhi^*] .
\end{equation*}
\end{Prop}
In the above Proposition the scalar solutions $\theta$ and $\theta^*$ of the linear problem  \eqref{eq:lin-dKP} and its adjoint \eqref{eq:lin-dKP-a} can be replaced by vectorial solutions $\Theta:\ZZ^N\to\DD^P$ and $\Theta^*:\ZZ^N\to(\DD^P)^T$, where the number of components of both column- and row-vectors is the same. We obtain then \emph{vectorial binary} (and also vectorial Darboux and adjoint Darboux) transformations. Vectorial binary transformations can be obtained as superposition of scalar binary Darboux transformations, which follows from the following observation. 
\begin{Prop}[\cite{TQL}]
\label{th:perm-fund}
Assume the following splitting of the data of the vectorial binary Darboux
transformation
\begin{equation*}
\Theta = \left( \begin{array}{c} 
\Theta^a \\ \Theta^b \end{array} \right),\qquad
\Theta^* = \left( \begin{array}{cc} 
\Theta_{a}^{*}  & \Theta_{b}^{*} \end{array} \right),
\end{equation*}
associated with the partition $\DD^P = \DD^{P_a} \oplus \DD^{P_b}$,
which implies the following splitting of the potentials
\begin{equation*} 
\Omega[\Theta,\bPhi^*] =  \left( \begin{array}{c} 
\Omega[\Theta^a,\bPhi^*] \\ 
\Omega[\Theta^b,\bPhi^*]  \end{array} \right), \quad
\Omega[\Theta,\Theta^*]  = \left( \begin{array}{cc}
\Omega[\Theta^a,\Theta^*_a]  &
\Omega[\Theta^a,\Theta^*_b]  \\
\Omega[\Theta^b,\Theta^*_a]  &
\Omega[\Theta^b,\Theta^*_b] \end{array} \right), 
\end{equation*}
\begin{equation*}
\Omega[\bPhi,\Theta^*]  = \left( \! \begin{array}{cc} 
\Omega[\bPhi,\Theta^*_a]   &
\Omega[\bPhi,\Theta^*_b] \end{array} \! \right).
\end{equation*} 
Then the vectorial binary Darboux transformation is equivalent to the following
superposition of vectorial binary Darboux transformations:\\
1) Transformation $\bPhi\to\bPhi^{\{a\}}$ and $\bPhi^*\to\bPhi^{*\{a\}}$ with the data 
$\Theta^a$, $\Theta_{a}^*$ and the corresponding
potentials $\Omega[\bPhi,\Theta^*_a]$, 
$\Omega[\Theta^a,\Theta^*_a]$, and
$\Omega[\Theta^a,\bPhi^*]$
\begin{align}
\label{eq:fund-vect-a}
\bPhi^{\{a\}} & = \bPhi -
\Omega[\bPhi,\Theta^*_a]
\Omega[\Theta^a,\Theta^*_a]^{-1}
\Theta^a_,
\\
\bPhi^{*\{a\}} & = \bPhi^* - \Theta^*_{a}
\Omega[\Theta^a,\Theta^*_a]^{-1}
\Omega[\Theta^a,\bPhi^*].
\end{align}
2) Application on the result the vectorial fundamental transformation with the
transformed data
\begin{align*}
{\Theta}^{b\{a\}} & = \Theta^b -
\Omega[\Theta^b,\Theta^*_a]
\Omega[\Theta^a,\Theta^*_a]^{-1}
\Theta^a,
\\
{\Theta}_{b}^{*\{a\}} & = \Theta_{b}^* - 
\Theta^*_{a}
\Omega[\Theta^a,\Theta^*_a]^{-1}
\Omega[\Theta^a, \Theta_{b}^*],
\end{align*}
and the potentials
\begin{align*}
{\Omega}[\Theta^b,\bPhi^*]^{\{a\}} & =
\Omega[\Theta^b,\bPhi^*] - 
\Omega[\Theta^b,\Theta^*_a]
\Omega[\Theta^a,\Theta^*_a]^{-1}
\Omega[\Theta^a,\bPhi^*]=
\Omega[{\Theta}^{b\{a\}},\bPhi^{*\{a\}}],
\\
{\Omega}[\Theta^b,\Theta^*_b]^{\{a\}} & =
\Omega[\Theta^b,\Theta^*_b] - 
\Omega[\Theta^b,\Theta^*_a]
\Omega[\Theta^a,\Theta^*_a]^{-1}
\Omega[\Theta^a,\Theta^*_b]=
\Omega[{\Theta}^{b\{a\}},{\Theta}_b^{*\{a\}}],
\\
{\Omega}[\bPhi,\Theta^*_b]^{\{a\}}  & =
\Omega[\bPhi,\Theta^*_b] - 
\Omega[\bPhi,\Theta^*_a]
\Omega[\Theta^a,\Theta^*_a]^{-1}
\Omega[\Theta^a,\Theta^*_b]=
\Omega[\bPhi^{\{a\}},{\Theta}_b^{*\{a\}}],
\end{align*}
which gives
\begin{align*} 
\bPhi^{\{a,b\}}  &= \bPhi^{\{a\}} - 
{\Omega}[\bPhi,\Theta^*_b]^{\{a\}}
[{\Omega}[\Theta^b,\Theta^*_b]^{\{a\}}]^{-1}
\Theta^{b\{a\}},\\
\bPhi^{*\{a,b\}}  &= \bPhi^{*\{a\}} - 
\Theta_b^{^*\{a\}}
[{\Omega}[\Theta^b,\Theta^*_b]^{\{a\}}]^{-1}
\Omega[{\Theta}^{b\{a\}},\bPhi^{*\{a\}}].
\end{align*}
\end{Prop} 
\begin{Rem}
The same final result we obtain starting from the binary transformation with the data
$\Theta^b$, $\Theta_{b}^*$ and suitably transforming the other part of the data.
\end{Rem}

\subsection{The $A$-type root lattice description of symmetries of the Hirota equation}
The definition of Desargues maps, as given above, does not take into account symmetries of the underlying Veblen and Desargues configurations; the symmetry of the Desargues configuration was discussed in relation to four dimensional consistency of the Hirota system (in its various forms) in \cite{Bobenko-talk}. A symmetric description of the Laplace sequences of two dimensional lattices of planar lattices, and the closely related discrete SKP equation in terms of the face centered cubic (FCC) lattice (which is the $Q(A_3)$ root lattice) is due to W. Schief~\cite{Schief-talk}. We generalized these results in \cite{Dol-AN} where the Desargues maps were studied from the point of view of the root lattices of type $A$ and of the corresponding affine Weyl groups. Our research was also motivated by relevance of the theory of root systems and Weyl groups in the theory of the discrete Painlev\'{e} equations \cite{Noumi,NY-AN,Sakai}.

\begin{figure}
\begin{center}
\includegraphics[width=8cm]{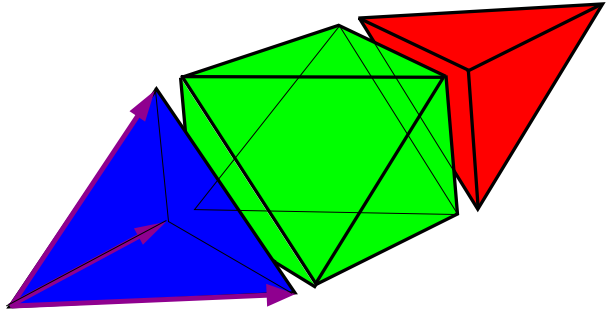}
\end{center}
\caption{A fundamental parallelogram (with the corresponding basis of vectors) of 
the $Q(A_3)$ root lattice 
decomposed into its Delaunay tiles: two tetrahedra $P(1,3)$ and $P(3,3)$, and
the octahedron $P(2,3)$.}
\label{fig:A3-cube}
\end{figure}
Recall~\cite{Bourbaki,ConwaySloane} that $N$-dimensional root lattice $Q(A_N)$
is generated by vectors along the edges of regular $N$-simplex. If we take the vertices of the simplex to be 
the vectors $\boldsymbol{e}_i = (0,\dots , \stackrel{i}{1}, \dots , 0)$ of the canonical basis in $\RR^{N+1}$ then the generators are 
$\boldsymbol{\varepsilon}^i_j = \boldsymbol{e}_i - \boldsymbol{e}_j$, 
$1\leq i \neq j \leq N+1$.
This identifies the root lattice as the set of all vectors $(m_1, \dots ,
m_{N+1}) \in \ZZ^{N+1}$ of integer coordinates with zero sum
$m_1 + \cdots + m_{N+1} = 0 $. The scalar product $(\cdot | \cdot)$ in the corresponding hyperspace
$\VV=\{ (x_1, \dots , x_{N+1}) \in \RR^{N+1}| \, x_1 + \cdots + x_{N+1} = 0 \}$
is inherited from the ambient $\RR^{N+1}$. The \emph{simple roots} of the lattice are the vectors $\boldsymbol{\alpha}_i = \boldsymbol{e}_{i} - \boldsymbol{e}_{i+1}$, 
$1\leq i \leq N$.

The \emph{holes} in the lattice are the points of $\VV$ that are
locally maximally distant from the lattice~\cite{ConwaySloane,MoodyPatera}. The convex hull of the lattice
points closest to a hole is called the \emph{Delaunay
polytope}. The Delaunay polytopes of the root lattice $Q(A_N)$
form a tessellation 
of $\VV$ by convex polytopes $P(k,N)$, $k=1,\dots,N$, see Figure~\ref{fig:A3-cube}. 
In particular, tiles of type $P(1,N)$ are congruent to the initial $N$-simplex generating the lattice. One can colour two dimensional facets (always triangles) of all Delaunay polytopes using two colours: a black triangle belongs to a tile $P(1,N)$, and a white triangle does not. In particular, all two dimensional facets of $N$-simplices $P(N,N)$ are white.

The \emph{Weyl group} $W_0(A_N)$ is the Coxeter group
generated by reflections $r_i$, $1\leq i \leq N$,
with respect of the
hyperplanes through the origin and orthogonal to the corresponding
simple roots 
\begin{equation} \label{eq:ri-action}
r_i:  \boldsymbol{v} \mapsto \boldsymbol{v} - 
2\frac{(\boldsymbol{v} | \boldsymbol{\alpha}_i )}
{(\boldsymbol{\alpha}_i | \boldsymbol{\alpha}_i )}
\boldsymbol{\alpha}_i.
\end{equation}
The group $W_0(A_N)$ is isomorphic to the symmetric group $S_{N+1}$ which act by
permuting the vectors $\boldsymbol{e}_i$, $1\leq i \leq N+1$; the generators
$r_i$ are identified then with the transpositions $\sigma_i=(i,i+1)$. The \emph{affine Weyl group} $W(A_N)$ is the Coxeter group
generated by $r_1, r_2, \dots , r_N$ and by an
additional affine reflection $r_0$
\begin{equation}\label{eq:r0-action}
r_0:  \boldsymbol{v} \mapsto \boldsymbol{v} - \left( 1 - 
2\frac{(\boldsymbol{v} | \tilde{\boldsymbol{\alpha}} )}
{(\tilde{\boldsymbol{\alpha}} | \tilde{\boldsymbol{\alpha}} )}\right) 
\tilde{\boldsymbol{\alpha}},
\end{equation}
where $\tilde{\boldsymbol{\alpha}}= 
\boldsymbol{e}_{1} - \boldsymbol{e}_{N+1}$ is the highest root. It is well known fact~\cite{MoodyPatera} that the affine Weyl group acts on the Delaunay tiling by permuting tiles within each class $P(k,N)$.

Let us introduce $\ZZ^N$ coordinates in the root lattice $Q(A_N)$ 
by the identification $\ZZ^N = \sum_{i=1}^N
\ZZ\boldsymbol{\varepsilon}^{N+1}_i= Q(A_N)$.
Then the Desargues maps can be redefined as maps $\phi:Q(A_N)\to\PP^M(\DD)$ such that the vertices of each basic $N$-simplex $P(1,N)$ are mapped into collinear points. 
Such a characterization exhibits the affine Weyl group symmetry of Desargues maps.
\begin{Th} \label{th:AW-sym}
If $\phi:Q(A_N)\to\PP^M$ is a Desargues map then also for any element $w$ 
of the affine Weyl group $W(A_N)$ the map $\phi \circ w$ is a Desargues map,
where we consider the natural action of $W(A_N)$ on the root lattice
$Q(A_N)$. 
\end{Th}
Before considering algebraic consequences (on the level of the non-commutative Hirota system) of the above geometric approach to Desargues maps let us discuss configurations of points and lines derived as images of particular tiles $P(k,L)$. First notice, that because the tile $P(2,3)$ has six vertices, and its faces are four black triangles and four white triangles, then its Desargues map image produces (see Figure~\ref{fig:Des-P23}) the configuration of six points and four lines --- each point is incident with two lines, and each line is incident with three points.
\begin{figure}
\begin{center}
\includegraphics[width=11cm]{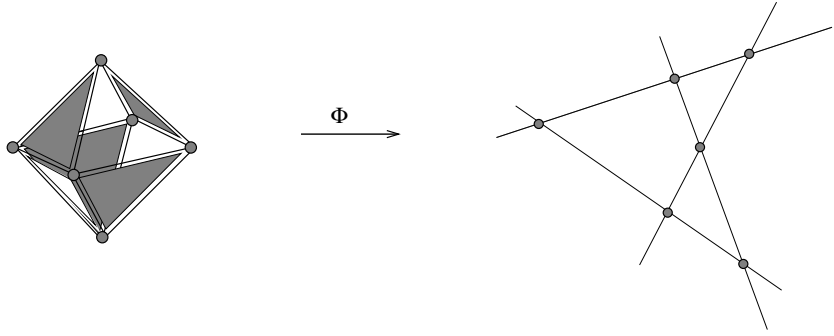}
\end{center}
\caption{The Veblen configuration $(6_2,4_3)$ as the image of the tile $P(2,3)$ under a Desargues map, and definition of the Laplace-Darboux maps \cite{Schief-talk}.}
\label{fig:Des-P23}
\end{figure}
We remark that such point of view was used~\cite{Schief-talk} by W.~Schief to provide a symmetric description of the Laplace sequence of two dimensional discrete conjugate nets~\cite{Sauer2,DCN}, which he called Laplace-Darboux maps of the FCC lattice; see also section~\ref{sec:D-QL}.

Next we consider image of a tile $P(3,4)$, where an analysis along the lines described in~\cite{MoodyPatera} demonstrates that it consists of: (i)~$10$ vertices, 
(ii)~$10$ black triangles (and no black tetrahedra) of which exactly $3$ meet at any vertex of the tile,
(iii)~$5$ facets $P(2,3)$.
Desargues maps produce then a configuration of (i)~$10$ points and (ii)~$10$ lines --- each line is incident with $3$ points, and each point is incident with $3$ lines, moreover (iii)~the configuration contains exactly $5$ Veblen configurations. It is known~\cite{Levi} that the last property select the Desargues configuration among other configurations satisfying the first two. 

Finally let us present algebraic counterpart of the geometric considerations on symmetries of Desargues maps and of the non-commutative Hirota system. Notice that the linear system \eqref{eq:lin-dKP} in the present notation reads (here $n\in Q(A_N)$ denotes an arbitrary point of the root lattice)
\begin{equation} \label{eq:lin-dKP-N+1}
\bPhi^{N+1}(n+\boldsymbol{\varepsilon}^{N+1}_i) - 
\bPhi^{N+1}(n+\boldsymbol{\varepsilon}^{N+1}_j) =  
\bPhi^{N+1}(n) 
U_{ij}^{N+1}(n),  \qquad 1\leq i \ne j \leq N,
\end{equation}
with the corresponding potentials $\rho^{N+1}_i$, $1=1,\dots,N$, such that
\begin{equation*}
U_{ij}^{N+1}(n) = \left[ \rho^{N+1}_i(n)\right]^{-1}
 \rho^{N+1}_i(n+\boldsymbol{\varepsilon}^{N+1}_j).
\end{equation*}
The upper index $N+1$ denotes a "sector", i.e. the choice of a basis along edges of the initial $N$-simplex of the lattice. Another basis
$\{ \boldsymbol{\varepsilon}^i_j \}$, where the index $i$ is fixed, and
$j\neq i$, is geometrically equivalent to the previous one, and should give a similar linear problem.
\begin{Prop} \label{prop:lin-i}
The functions $\bPhi^{i}\colon Q(A_N) \to \DD^{M+1}$
given by 
\begin{equation} \label{eq:phi-i}
\bPhi^{i}(n) = (-1)^{(n|\boldsymbol{\varepsilon}^{N+1}_i)}
\bPhi^{N+1}(n)
\left[ \rho^{N+1}_i(n)\right]^{-1},
\end{equation}
satisfy the linear system in the $i$-th sector
\begin{equation} \label{eq:lin-dKP-i}
\bPhi^{i}(n+\boldsymbol{\varepsilon}^{i}_j) - 
\bPhi^{i}(n+\boldsymbol{\varepsilon}^{i}_k) =  
\bPhi^{i}(n) 
U_{jk}^{i}(n),  \qquad i, j , k \quad \text{distinct,}
\end{equation}
where
\begin{equation}
U_{jk}^{i}(n) = \left[ \rho^{i}_j(n)\right]^{-1}
 \rho^{i}_j(n+\boldsymbol{\varepsilon}^{i}_k),
\end{equation}
and the potentials $\rho^{i}_j$ are given by
\begin{equation} \label{eq:rho-ij}
\rho^{i}_j(n) = \begin{cases}
\rho^{N+1}_j(n) \left[ \rho^{N+1}_i(n)\right]^{-1} , & \qquad j \neq N+1,\\
\left[ \rho^{N+1}_i(n)\right]^{-1} , & \qquad j = N+1.
\end{cases}
\end{equation}
In consequence the functions $U^\ell_{ij}$ satisfy the non-commutative Hirota system
\begin{align} \label{eq:alg-comp-U-l}
& U^\ell_{ij}(n) + U^\ell_{ji}(n) = 0, \qquad  
U^\ell_{ij}(n) + U^\ell_{jk}(n) + U^\ell_{ki}(n) = 0,\\
& \label{eq:U-rho-l} 
U^\ell_{kj}(n)U^\ell_{ki}(n+\boldsymbol{\varepsilon}^\ell_j) = 
U^\ell_{ki}(n) U^\ell_{kj}(n+\boldsymbol{\varepsilon}^\ell_i),
\end{align}
where the index $\ell$ is fixed, and all other indices $i,j,k$ are different from $\ell$, mutually distinct and, range from $1$ to $N+1$.
\end{Prop}
Notice that 
changing the sector can be understood as a symmetry transformation of
both the
linear and nonlinear systems. Other generators of symmetries are 
translations and permutations of indices within a fixed sector. 

Before we describe the affine Weyl group symmetries of the non-commutative Hirota system we observe that
\begin{equation} \label{eq:ccc}
\rho^i_j(n) \rho^k_i(n) = \rho^k_j(n), \qquad \text{where by definition}
\quad \rho^i_i = 1.
\end{equation}
It is natural to associate $\rho^i_j(n)$ with the oriented edge $[n, n+ \boldsymbol{\varepsilon}^{i}_{j}]$ of the root lattice, and define the action of the affine Weyl group on the functions $\rho^i_j$ through its action on the edges. In particular, the actions of the generators $r_i$, $i=0,\dots ,N$, of the affine Weyl group is given by 
\begin{equation} \label{eq:ri-rij}
(r_i . \rho^j_k)(n) = \rho^{\sigma_i(j)}_{\sigma_i(k)}(r_i(n)),
\end{equation}
where $\sigma$'s are the transpositions
$\sigma_i = (i, i+1)$, $i=1,\dots ,N$, and $\sigma_0 = (1,N+1)$. The action on the potentials $U^i_{jk}$ follows from action on the potentials $\rho^i_j$.

We distinguish the simple root functions 
$\rho^i = \rho^i_{i+1}$, $i=1, \dots, N$, which are attached to the simple 
roots directions. 
One can check, using the condition \eqref{eq:ccc}, that for $i< j$ we have
\begin{equation} \label{eq:r-ij}
\rho^i_j = \rho^{j-1} \dots \rho^i, \qquad \rho^j_i = (\rho^i_j)^{-1}.
\end{equation}
Let us define also the function $\rho^0$ as attached to the direction of the
root $\boldsymbol{\alpha}_0 = -\tilde{\boldsymbol{\alpha}}$,
which by \eqref{eq:r-ij} gives
\begin{equation}
\rho^0 = (\rho^N \rho^{N-1} \dots \rho^1)^{-1}.
\end{equation}
\begin{Prop}
The action of the generators $r_i$, $i=0,\dots ,N$, of the affine Weyl group on
the functions $\rho^j$, $j=0,\dots ,N$, is given by 
\begin{equation}
(r_i . \rho^j)(n) = [(\rho^i)^{-a_{ji}^U}\rho^j (\rho^i)^{-a_{ji}^L}](r_i(n)),
\end{equation}
where $a_{ji}^U$ and $a_{ji}^L$ are the "upper" and the "lower" parts
of the Cartan matrix of the affine Weyl group $W(A_N)$
\begin{equation}
a_{ij}^L = \left[ \begin{array}{rrrrr} 1 & 0 &  &  & -1 \\
-1 & 1 & 0 &  &  \\
 & -1 & 1 &  \ddots &     \\
 & & \ddots& \ddots & 0 \\
0 & & & -1 & 1 
\end{array} \right], \qquad
a_{ij}^U = \left[ \begin{array}{rrrrr} 1 & -1 &  &  & 0 \\
0 & 1 & -1 &  &  \\
 & 0 & 1 &  \ddots  &    \\
 & & \ddots &\ddots & -1 \\
-1 &  & & 0 & 1 
\end{array} \right].
\end{equation}
\end{Prop}
\begin{Rem}
It is remarkable \cite{ABS-oct}, compare also \cite{Bobenko-talk}, that in the commutative case the Hirota equation (in its various forms) is the only one system of the octahedron type multidimensionally consistent on the root lattice of type $A$.
\end{Rem}
\subsection{Desargues maps and multidimensional quadrilateral lattices}
\label{sec:D-QL} 
Let $N=2K-1$ be odd, we will first describe $\ZZ^K = Q(B_K)$ root sublattices of $Q(A_{2K-1})$. We split the standard basis vectors $(\boldsymbol{e}_i)_{i=1}^{2K}$ into $K$ pairs (we order pairs and elements within the pairs). We fix the splitting $( \boldsymbol{e}_1, \boldsymbol{e}_2)$, \dots  $( \boldsymbol{e}_{2K-1}, \boldsymbol{e}_{2K})$, any other splitting can be obtained by action of the symmetric group $S_{2K}=W_0(A_{2K-1})$. 
Define vectors $\boldsymbol{E}_i = \boldsymbol{e}_{2i-1} - \boldsymbol{e}_{2i}$, $i=1,\dots ,K$, which satisfy the orthogonality relations
\begin{equation*}
( \boldsymbol{E}_i | \boldsymbol{E}_j ) = 2 \delta_{ij}, 
\end{equation*}
and generate the $\ZZ^K$ sublattice (with rescaled standard scalar product) in the root lattice $Q(A_{2K-1})$. Notice that the change of orientation of the vector $\boldsymbol{E}_i$ (which means transposition of the corresponding pair) or transposition of the vectors $\boldsymbol{E}_i$ and $\boldsymbol{E}_j$ (which means corresponding permutation of pairs) are automorphisms of the sublattice. Therefore we have $(2K)!/ (2^K K!) = (2K-1)!!$ different such sublattices (the origin is fixed) in $Q(A_{2K-1})$. 

The above mentioned automorphisms of the lattice can be described as reflections in (the hyperplanes orthogonal to) the vectors $\boldsymbol{E}_i$ or to the vectors $\boldsymbol{E}_i - \boldsymbol{E}_j$. The minimal (in order to obtain all above mentioned symmetries which generate the group $W_0(B_K) = \ZZ_2^K \rtimes S_K$) set of such vectors is the standard simple root system of type $B_K$
\begin{equation}
\boldsymbol{E}_1 - \boldsymbol{E}_2, \quad \boldsymbol{E}_2 - \boldsymbol{E}_3, \dots , \quad \boldsymbol{E}_{K-1}- \boldsymbol{E}_K, \quad \boldsymbol{E}_K.
\end{equation}

We are now in a position to describe relation of Desargues maps to the multidimensional quadrilateral lattice maps~\cite{DCN,MQL}, which are maps $\psi\colon\ZZ^K\to\PP^M(\DD)$ characterized by the condition that for arbitrary $m\in\ZZ^K$ and for any pair of indices $i\neq j$ the points $\psi(m)$, $\psi(m+\boldsymbol{E}_i)$, $\psi(m+\boldsymbol{E}_j)$ and $\psi(m+\boldsymbol{E}_i+\boldsymbol{E}_j)$ are coplanar. Instead of the Hirota system (or discrete modified KP system) we have then, for $K\geq 3$, the so called discrete Darboux equations~\cite{BoKo}. The Desargues map condition of collinearity of three points may be considered as degeneration of the quadrilateral map condition of coplanarity of four points. However it turns out~\cite{Dol-Des} that both objects are almost equivalent, as described below.
\begin{Th}
Given Desargues map $\phi\colon Q(A_{2K-1})\to \PP(\DD)^M$, then its restriction to $Q(B_K)$ sublattice is a quadrilateral lattice map, and the Desargues map can be recovered from its quadrilateral map restriction. 
\end{Th}
\begin{proof}
Let us introduce the following change of $\ZZ^{2K-1}$ coordinates in the lattice $Q(A_{2K-1})$
\begin{equation*}
n = \sum_{i=1}^{2K-1} n_i \boldsymbol{\varepsilon}^{2K}_i =
-\sum_{j=1}^K m_j \boldsymbol{E}_j - \sum_{j=1}^K \ell_j 
\boldsymbol{\varepsilon}^{2K}_{2j}, \qquad \text{i.e.} \quad
n_{2j-1} = m_j, \quad n_{2j} = - m_j - \ell_j, \qquad 1\leq j\leq K,
\end{equation*}
where we also defined $n_{2K} = -(n_1 + \ldots + n_{2K-1})$, which implies $\ell_1 + \ldots + \ell_K =0$. We have therefore $m= \sum_{j=1}^K m_j \boldsymbol{E}_j \in Q(B_K)=\ZZ^K$, and $\ell = (\ell_1, \ldots, \ell_K) \in Q(A_{K-1})$. 

For fixed $\ell\in Q(A_{K-1})$ define the
map $\psi^\ell:\ZZ^K\to \PP^M$ given by
$\psi^\ell(m) = \phi(n)$, where the relation between $n$ and $m$ and $\ell$ is
given above. For $\ell = (0,\ldots,0)$ we obtain therefore the restriction of our Desargues map $\Phi$ to the root lattice $Q(B_K)$ which we would like to show is a quadrilateral lattice map. We will demonstrate that for arbitrary $\ell\in Q(A_K)$, which we will use in the next part of the proof. 

To show planarity of elementary quadrilaterals in variables $m_i$ and $m_K$, $i\neq K$, notice that upon identification of $\psi^\ell$ with $\phi$ the point $\psi^\ell_{(K)}$ should be identified with $\phi_{(2K-1)}$. Moreover the point
$\psi^{\ell - \boldsymbol{e}_i +\boldsymbol{e}_j}_{(i)}$, where $\boldsymbol{e}_i$ is the vector of the canonical basis in $\ZZ^K\supset Q(A_{K-1})$,
should be identified with $\phi_{(2i-1)}$. Because the points $\phi$, $\phi_{(2K-1)}$ and $\phi_{(2i-1)}$ are collinear from the very definition of Desargues maps we have therefore collinearity of $\psi^\ell$, $\psi^\ell_{(K)}$ and 
$\psi^{\ell - \boldsymbol{e}_i +\boldsymbol{e}_j}_{(i)}$ (see Figure~\ref{fig:D-QL-ij}). Similarly, the point $\psi^\ell_{(i)}$ should be identified with $\phi_{(2i-1,-2i)}$, and the point $\psi^\ell_{(ij)}$ with $\phi_{(2i-1,-2i,2K-1)}$. Again we have that $\psi^\ell_{(i)}$, $\psi^\ell_{(ij)}$ and $\psi^{\ell - \boldsymbol{e}_i +\boldsymbol{e}_j}_{(i)}$ are collinear, which together with the previous collinearity imply coplanarity of the points $\psi^\ell$, $\psi^\ell_{(K)}$, $\psi^\ell_{(i)}$ and $\psi^\ell_{(ij)}$.
\begin{figure}
\begin{center}
\includegraphics[width=16cm]{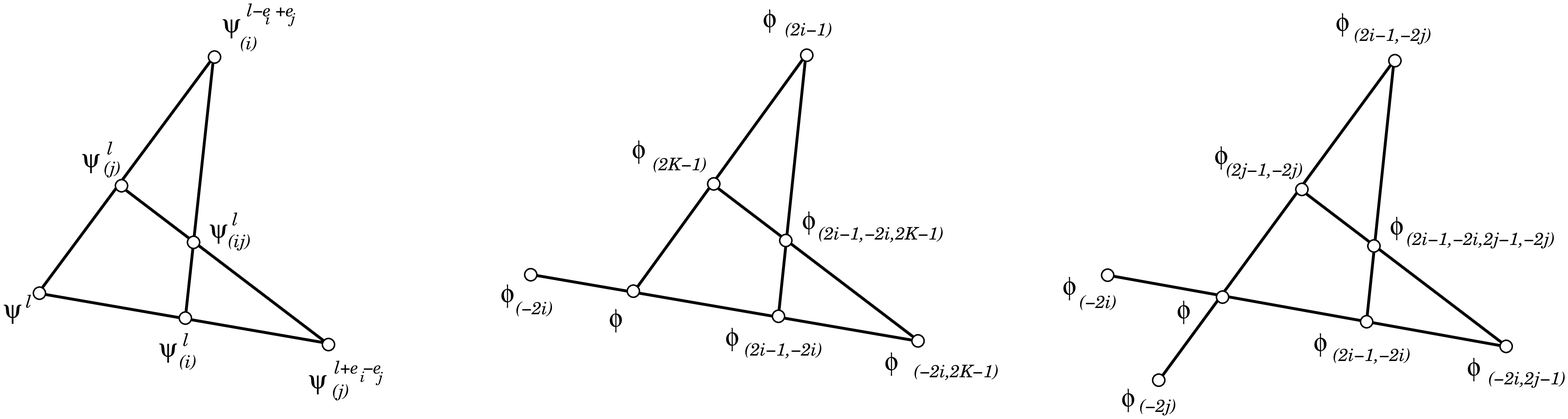}
\end{center}
\caption{Planarity of elementary quadrilaterals of $Q(A_{2K-1})$ Desargues maps restricted to $Q(B_K)$ root lattices.}
\label{fig:D-QL-ij}
\end{figure}
The corresponding calculation for variables $m_i$ and $m_j$, where $1\leq i \neq j < K$, was performed in \cite{Dol-Des}, and here we just refer to Figure~\ref{fig:D-QL-ij}. 

The transformation of the map $\psi^\ell\colon\ZZ^K\to \PP^M$ into 
$\psi^{\ell+\boldsymbol{e}_i -\boldsymbol{e}_j}\colon \ZZ^K\to \PP^M$ is called the Laplace transformation $\mathcal{L}_{ij}$ \cite{DCN,TQL}. Geometrically it is given by intersection of opposite tangent lines of planar quadrilaterals, as visualized in Figure~\ref{fig:D-QL-ij}. By iterative application to the restricted quadrilateral lattice map it allows to reconstruct the whole Desargues map $\phi$.
\end{proof}
\begin{Cor}
Any quadrilateral lattice map can be extended in the above described way to the corresponding Desargues map.
\end{Cor}
\begin{proof}
It is known~\cite{DCN,TQL} that the Laplace transformations $\mathcal{L}_{ij}$ can be defined for generic quadrilateral lattice maps giving new quadrilateral lattice maps. Moreover, they satisfy relations
\begin{equation*}
\mathcal{L}_{ij} \circ \mathcal{L}_{ji}  = \text{id} \; ,\qquad
\mathcal{L}_{jk} \circ \mathcal{L}_{ij} = \cL_{ik}, \qquad
\mathcal{L}_{ki} \circ \mathcal{L}_{ij} = \cL_{kj}, 
\end{equation*}
which allows to parametrize the
quadrilateral lattices generated from one quadrilateral lattice via the Laplace
transformations by points of the root lattice $Q(A_{K-1})$  \cite{DMMMS} and introduce, in addition to the initial $m\in \ZZ^K$ variable of the quadrilateral lattice map, the new $\ell\in Q(A_{K-1})$ variable. After the change to the $n\in Q(A_{2K-1})$ variable as above one obtains from an arbitrary quadrilateral lattice map the corresponding Desargues map.
\end{proof}
At this point we may study the quadrilateral lattice maps and the corresponding discrete Darboux equations, their Darboux-type transformations and reductions; apart from above cited works see also \cite{Bobenko-O,KoSchief2,AKV,DS-EMP,Shapiro} for geometric but also analytic (in the case of the field of complex or real numbers) tools to study such maps and corresponding solutions of the discrete Darboux system. We remark that the pioneering works of A.~Bobenko and U.~Pinkall and their collaborators on discrete isothermic surfaces~\cite{BP2,HeHP,Schief-C} can be directly incorporated in the theory of multidimensional lattices of planar quadrilaterals~\cite{BobSur-Proc,Doliwa-isoth}. Also discrete pseudospherical surfaces~\cite{BP1} together with more general discrete asymptotic surfaces~\cite{Sauer} can be considered as reductions of quadrilateral lattices indirectly via the Pl\"{u}cker embedding \cite{W-cong}; see also other related works~\cite{KonopelchenkoPinkall,BobenkoSchief,Schief-ass,Nieszporski,DNS-ass,Schroecker}. Investigation of quadrilateral lattices in the general division ring context has been initiated in~\cite{gaql}. For the relation of the quadrilateral lattice maps with Zamolodchikov's tetrahedron equation~\cite{Zamolodchikov} and related integrable quantum field models, which was an important motivation for results presented in the next section, see~\cite{BaMaSe,Sergeev-q3w}.

\section{The quantum plane structure maps}
\label{sec:CF}
In the previous part of the paper no commutativity on the level of the dependent variables was assumed. Obviously, all results obtained there are valid for commutative dependent variables as well, where also additional techniques are available. See~\cite{KWZ} for application of the algebro-geometric techniques and \cite{Dol-Des} for application of the non-local $\bar\partial$-dressing method in the complex field case, and \cite{BD} for modification of results of \cite{Krichever-difference,KWZ} to the finite field case. 

As it was observed in \cite{DoliwaSergeev-pentagon} the Desargues map equations~\eqref{eq:alg-cond}-\eqref{eq:shift-cond}, which in particular gauge give the non-commutative Hirota system \eqref{eq:alg-comp-U}-\eqref{eq:U-rho}, can be decomposed into two maps which are solutions of the functional pentagonal equation. Recall that a map $W:\XX\times \XX \to \XX\times \XX$ satisfies the functional (or set-theoretical) pentagon equation \cite{Zakrzewski} if
\begin{equation} \label{eq:pentagon-W-X}
W_{12} \circ W_{23} = W_{23} \circ  W_{13} \circ W_{12}, \qquad \text{in} \quad \XX \times \XX \times \XX,
\end{equation}
regarded as an equality of composite maps; here $W_{ij}$ acts as $W$ in $i$-th and $j$-th factors of the Cartesian product. 

Another observation made in \cite{DoliwaSergeev-pentagon} was that the ultra-locality restriction imposed on these maps leads to Weyl commutation relations and gives rise to the corresponding solutions of the quantum pentagon equation~\cite{BaajSkandalis}. Guided by the application of the quantum pentagon equation in the quantum group theory we will exploit the pentagonal maps related to the Hirota system in order to define a bialgebra structure in the quantum plane algebra. 

\subsection{The normalization map}
\subsubsection{Pentagon property of the normalization map}
The first part~\eqref{eq:alg-cond} of the Desargues map equations relates coefficients of linear equations~\eqref{eq:lin-A} of collinear points of nearest positive neighbours of $\phi$. We study in more detail such a relation for 
four distinct collinear points $A$, $B$, $C$, and $D$ (compare Figures~\ref{fig:Desargues-Veblen-normal} and~\ref{fig:square-lin-ABCD-W-n}). 
\begin{figure}
\begin{center}
\includegraphics[width=8cm]{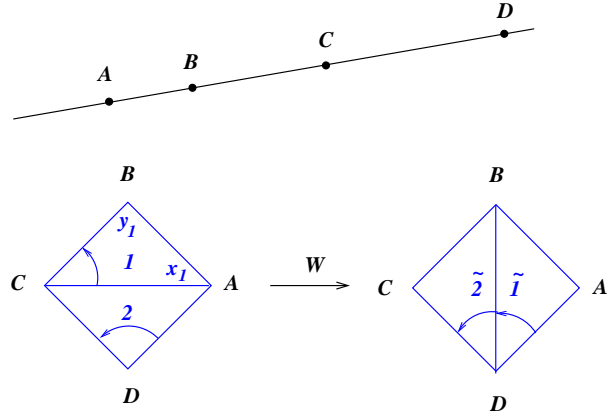}
\end{center}
\caption{Graphic representation of the linear relations for the normalization map $W$}
\label{fig:square-lin-ABCD-W-n}
\end{figure} 
Denote by $\bPhi_I$, $I=A,B,C,D$ corresponding homogeneous coordinates of the points and consider the following two pairs of linear relations
\begin{align*}
\bPhi_C = & \bPhi_A x_1  +  \bPhi_B  y_1, & 
\bPhi_D = &  \bPhi_A \tilde{x}_1 +  \bPhi_B \tilde{y}_1 , \\
\bPhi_D = & \bPhi_A  x_2 + \bPhi_C  y_2 , & 
\bPhi_D = &  \bPhi_B \tilde{x}_2 +  \bPhi_C  \tilde{y}_2. & 
\end{align*}
They define a birational map 
\begin{equation*}
W:\DD^2\times \DD^2 \ni ((x_1,y_1), (x_2, y_2)) \dashrightarrow 
((\tilde{x}_1,\tilde{y}_1), (\tilde{x}_2, \tilde{y}_2)) \in \DD^2\times \DD^2,
\end{equation*}
given explicitly by 
\begin{equation} \label{eq:W}
\tilde{x}_1 =  x_2 +  x_1 y_2, \qquad \tilde{y}_1 =   y_1 y_2,\qquad
\tilde{x}_2 =  -  y_1 x_1^{-1} x_2 , \qquad \tilde{y}_2 =  y_2 + x_1^{-1} x_2 .
\end{equation} 
We represent graphically the map $W$ as follows: linear relations between homogeneous coordinates of triplets of collinear points are represented by triangles. The information about the $x$ and $y$ coefficients of a linear equation is encoded in the arrow which starts at the edge pointing the $x$-coefficient point and ends at the edge pointing the $y$-coefficient point. The vertex where the arrow is placed represent the point with the coefficient equal to one, see Figure~\ref{fig:square-lin-ABCD-W-n}. 

Let us present an important property of the map $W$. It can be best seen if we add a fifth collinear point $E$ to the previous four points. We start from three linear relations visualized on Fig.~\ref{fig:pentagon} and perform the transformations according to the geometric description of the map as described on Fig.~\ref{fig:square-lin-ABCD-W-n}. There are two ways to change the linear systems, as illustrated on Fig.~\ref{fig:pentagon}, and both give the same result, which can be formulated as follows. 
\begin{figure}
\begin{center}
\includegraphics[width=13cm]{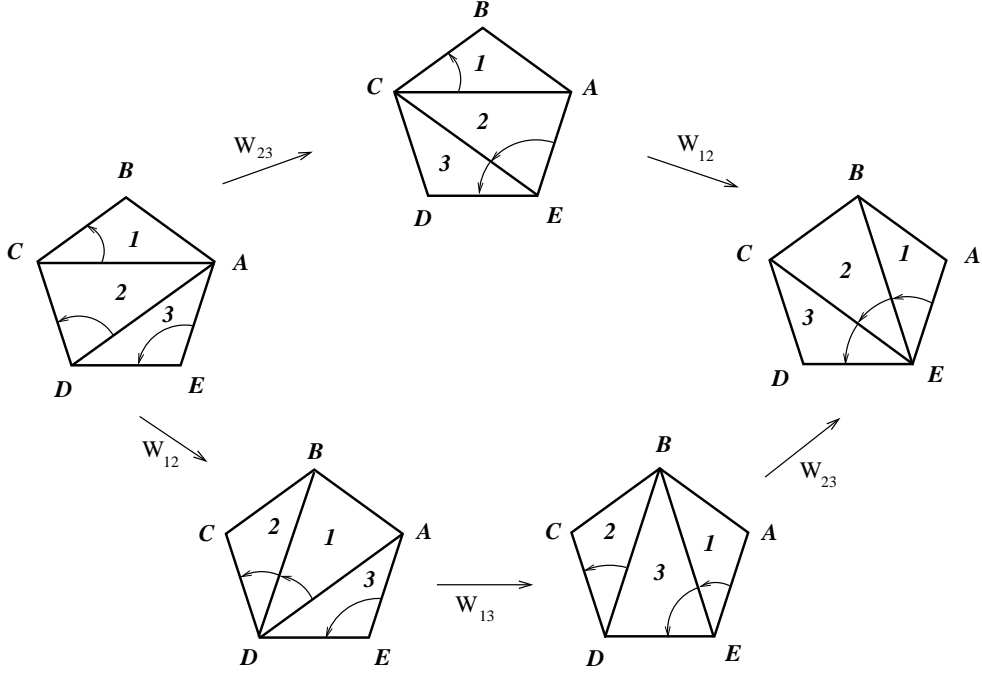}
\end{center}
\caption{The pentagon property of the map $W$}
\label{fig:pentagon}
\end{figure} 
\begin{Prop} \label{prop:W-pent}
The map $W:\DD^2\times \DD^2 \ni ((x_1,y_1), (x_2, y_2)) \dashrightarrow 
((\tilde{x}_1,\tilde{y}_1), (\tilde{x}_2, \tilde{y}_2)) \in \DD^2\times \DD^2 $ given by equations \eqref{eq:W} satisfies the functional pentagon equation \eqref{eq:pentagon-W-X}.
\end{Prop}
\begin{proof}
Those who are not convinced by the above geometric arguments can check directly that
\begin{equation*}
W_{12} \circ W_{23} \left( \begin{matrix}
x_1 , y_1 \\ x_2 ,  y_2 \\ x_3 ,  y_3
\end{matrix} \right) = 
\left( \begin{array}{rl}
x_3 + x_2 y_3 + x_1 y_2 y_3, & y_1 y_2 y_3 \\ 
- y_1 x_1^{-1} (x_3 + x_2 y_3 ) , & y_2 y_3 + x_1^{-1} (x_3 + x_2 y_3 )  \\ 
- y_2 x_2^{-1} x_3 , & y_3 + x_2^{-1} x_3 
\end{array} \right) 
= W_{23} \circ W_{13} \circ W_{12}
\left( \begin{matrix}
x_1, y_1 \\ x_2, y_2 \\ x_3, y_3
\end{matrix} \right).
\end{equation*}
\end{proof}

\subsubsection{The ultra-local reduction of the normalization map}
Let us assume that the elements $x_i$, $y_i$ with different indices commute
\begin{equation} \label{eq:UL-xy}
x_i x_j = x_j x_i, \qquad y_i y_j = y_j y_i, \qquad x_i y_j = y_j x_i, \qquad i\neq j,
\end{equation}
We are interested in the situation where the same holds for output elements 
$\tilde{x}_i$, $\tilde{y}_i$ of the map $W$. Denote by $\Bbbk\in\mathcal{Z}(\DD)$ a fixed subfield of the center of the division ring, and let us make few technical (general position) assumptions:\\
1) The elements $x_i$, $y_i$, $i=1,2$, do not belong to the field $\Bbbk$ and are linearly  independent (as elements of the $\Bbbk$-vector space $\DD$).\\
2) The intersection of division hulls $\DD^{(i)}$ generated by elements $x_i,y_i$, $i=1,2$, is the field $\Bbbk$ only.

The following result was obtained in \cite{DoliwaSergeev-pentagon} for the inverse of the map $W$.
\begin{Prop}
If the normalization map $W$ preserves the ultra-locality conditions \eqref{eq:UL-xy} then, under the above general position conditions, there exists a $q\in\Bbbk^\times$ such that the Weyl commutation relations hold
\begin{equation} \label{eq:WCR-xy}
x_i y_i = q y_i x_i, \qquad i=1,2.
\end{equation}
\end{Prop}
\begin{proof}
Commutation between $\tilde{x}_1$ and $\tilde{x}_2$ leads directly to equality
\begin{equation*}
x_1 y_1 x_1^{-1} y_1^{-1} = x_2 y_2 x_2^{-1} y_2^{-1} ,
\end{equation*}
which by the second part of the general position conditions gives the statement. Other three ultra-locality conditions for the output elements give the same result or are trivially satisfied.
\end{proof}

Denote by $\Bbbk_q[x,y]$ the quantum Manin plane~\cite{Manin,Kassel} generated by indeterminates $x,y$ satisfying the Weyl $q$-commutation relations $xy = qyx$. 
The $\Bbbk$-subalgebra $\Bbbk_q[x_1,y_1,x_2,y_2]$ of division ring $\DD$ is isomorphic to $\Bbbk_q[x,y]^{\otimes 2}$. It is well known~\cite{BrownGoodearl} that both $\Bbbk_q[x,y]$ and $\Bbbk_q[x,y]^{\otimes 2}$ have division algebras of fractions (called algebras of quantum rational functions), denoted by $\Bbbk_q(x,y)$ and $\Bbbk_q(x_1,y_1,x_2,y_2)$, correspondingly.

It is easy to check that the output elements satisfy not only the ultra-locality condition (which we assumed) but also the map preserves the Weyl commutation relations.
\begin{Cor} \label{cor:W-aut}
The map $W$ provides automorphism of  the division algebra $\Bbbk_q(x_1,y_1,x_2,y_2)$.
\end{Cor}

In the quasiclassical limit $q\to 1$ the $q$-commutation relations 
are replaced by the Poisson algebra structure in the field $\Bbbk(x_1,y_1,x_2,y_2)$ of rational functions of four variables, with the bracket given by
\begin{equation} \label{eq:P-b}
\{ x_i, y_i \} = x_i y_i, \qquad \{ x_i, x_j \} = 0, \qquad \{ y_i, y_j \} = 0,
\qquad \{ x_i, y_j \} =0, \qquad i\neq j .
\end{equation}
In consequence, the normalization map $W$ is also a Poisson automorphism of the field $\Bbbk(x_1,y_1,x_2,y_2)$.

\subsubsection{The bialgebra structure of the quantum plane}
In this Section we use the identification of $\Bbbk_q[x_1,y_1,x_2,y_2]$ with $\Bbbk_q[x,y]\otimes \Bbbk_q[x,y]$, i.e. 
\begin{equation*}
x_1 = x\otimes 1, \qquad y_1 = y\otimes 1, 
\qquad x_2 = 1\otimes x, \qquad y_2 = 1\otimes y.
\end{equation*}
We use also the embedding of $\Bbbk_q[x,y]$ into $\Bbbk_q[x,y]\otimes \Bbbk_q[x,y]$ as the first factor, generated by
\begin{equation*}
x \mapsto x\otimes 1, \qquad y \mapsto y \otimes 1.
\end{equation*}
Then, by Corollary~\ref{cor:W-aut}, the map $W$ allows to define unital homomorphism of algebras $\Delta:\Bbbk_q[x,y] \to \Bbbk_q[x,y]\otimes \Bbbk_q[x,y]$, given on generators by
\begin{equation}
\Delta (x) = 1\otimes x + x \otimes y, \qquad \Delta (y) = y \otimes y.
\end{equation}
Such a $\Delta$ is coassociative
\begin{equation*}
[(\Delta \otimes \mathrm{id} ) \circ \Delta ] \left( \begin{matrix} x \\ y  
\end{matrix} \right) =
\left( \begin{matrix} 1 \otimes 1 \otimes x + 1 \otimes x \otimes y +
x \otimes y \otimes y \\   y \otimes y \otimes y \end{matrix}  \right)
= [( \mathrm{id} \otimes \Delta) \circ \Delta ]\left( \begin{matrix} x \\ y  
\end{matrix} \right) ,
\end{equation*}
which can be checked directly, but actually it is a consequence of Proposition~\ref{prop:W-pent}.
\begin{Rem}
In~\cite{DoliwaSergeev-pentagon} we represented in a special case the (inverse of the present) map $W$ by an inner automorphism $\boldsymbol{W}$ of suitably completed $\Bbbk_q(x_1,y_1,x_2,y_2)$, which up to appropriate scaling satisfies quantum pentagon equation 
\begin{equation*}
\boldsymbol{W}_{23} \boldsymbol{W}_{12} = \boldsymbol{W}_{12} \boldsymbol{W}_{13} \boldsymbol{W}_{23} . 
\end{equation*}
The inner authomorphism operator $\boldsymbol{W}$ can be constructed in terms of the so called non-compact quantum dilogarithm function \cite{Faddeev-LMP,ax+b,Kashaev-P}.
Then the coproduct $\Delta$ is defined by equation \cite{BaajSkandalis}
\begin{equation*}
\Delta(a) = \boldsymbol{W} (a\otimes 1) \boldsymbol{W}^{-1}.
\end{equation*}
\end{Rem}

Given the coproduct $\Delta$ we can find in the standard way (see for example \cite{Sweedler,Kassel}) the other structure maps. In looking for a unital homomorphism $\epsilon:\Bbbk_q[x,y] \to \Bbbk$ (the counit) compatible with the coproduct 
\begin{equation*}
(\epsilon  \otimes \mathrm{id} ) \circ \Delta = \mathrm{id} = ( \mathrm{id} \otimes \epsilon) \circ \Delta,
\end{equation*}
we find after simple calculation
\begin{equation*}
\epsilon(y) = 1, \qquad \epsilon(x) = 0.
\end{equation*}
In this way we completed derivation of a bialgebra structure of the quantum plane.

The above bialgebra can be extended to the Hopf algebra structure, provided we enlarge $\Bbbk_q[x,y]$ to $\Bbbk_q[x,y,y^{-1}]$. 
The antipode antihomomorphism $S:\Bbbk_q[x,y,y^{-1}] \to \Bbbk_q[x,y, y^{-1}]$ can be derived from the compatibility condition
\begin{equation*}
\sum_{(a)} S (a_{(1)} ) a_{(2)} = \sum_{(a)} a_{(1)} S ( a_{(2)}) = \epsilon(a) 1 , \qquad 
\textrm{where} \qquad \Delta(a) = \sum_{(a)} a_{(1)} \otimes a_{(2)},
\end{equation*}
and is given on the generators by
\begin{equation*}
S(y) = y^{-1}, \qquad S(x) = -x y^{-1}.
\end{equation*}
This is the standard Hopf algebra structure on the quantum group of affine transformations of the line \cite{BaajSkandalis}, see also \cite{Sweedler} for the free algebra case.

\subsection{The Veblen flip and its pentagonal property}
\subsubsection{Symmetry of the Desargues configuration and the Veblen flip}
As it was explained in \cite{DoliwaSergeev-pentagon} the second part \eqref{eq:shift-cond} of the Desargues map system describes the Veblen configuration (all the points on Figure~\ref{fig:Desargues-Veblen-normal} except of $\phi$). To study the Veblen configuration it is convenient to exploit its symmetry, and we label its points by two (different) letters out of four. Then the lines are labelled by three letters, and the incidence relation is defined by containment, see Figure~\ref{fig:Veblen-flip-tetr}. 
\begin{figure}
\begin{center}
\includegraphics[width=14cm]{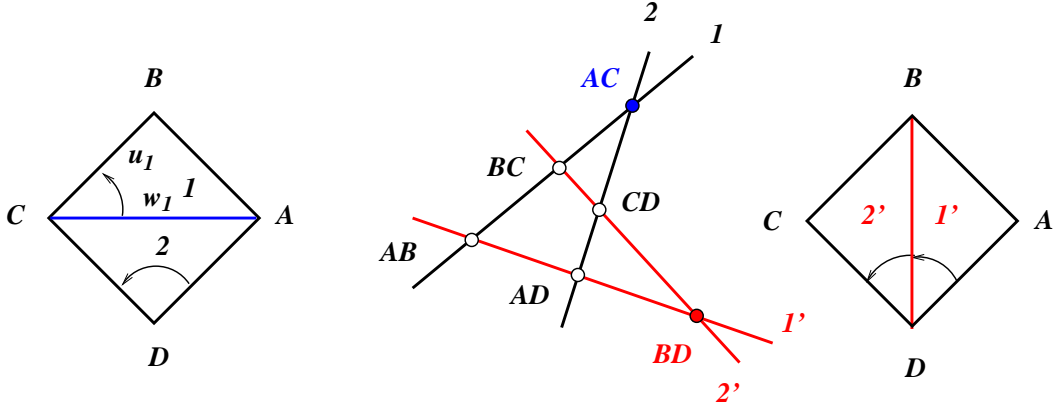}
\end{center}
\caption{The Veblen flip and its simplex representation}
\label{fig:Veblen-flip-tetr}
\end{figure}
Given five labelled points of the configuration, which belong to two initial lines, we can uniquely determine the sixth point as the intersection of two new lines. Algebraically, we start from two linear relations, which is convenient to chose in the form
\begin{align*} 
\bPhi_{AB} = &  \bPhi_{AC}  w_1 +   \bPhi_{BC} u_1 ,\\
\bPhi_{AC} = &  \bPhi_{AD}  w_2 +   \bPhi_{CD} u_2 .
\end{align*}
After the Veblen flip two new lines intersecting in the new point $BD$ give the relations
\begin{align*}
\bPhi_{AB} = &  \bPhi_{AD}  w^\prime_1 + \bPhi_{BD} u^\prime_1 ,\\
\bPhi_{BC} = & \bPhi_{BD}   w^\prime_2 + \bPhi_{CD} u^\prime_2  ,
\end{align*}
which define the map $W^G:\DD^2\times \DD^2 \ni ((u_1,w_1), (u_2, w_2)) \dashrightarrow 
((u^\prime_1,w^\prime_1), (u^\prime_2 , w^\prime_2)) \in \DD^2\times \DD^2 $
\begin{equation} \label{eq:WG}
u^\prime_1 = G w_1, \qquad w^\prime_1= w_2 w_1, \qquad
u^\prime_2 = -u_2 w_1 u_1^{-1}, \qquad w^\prime_2 = G w_1 u_1^{-1},
\end{equation}
where the free parameter $G$ represents allowed scaling in definition of homogeneous coordinates of the new point
\begin{equation*}
\bPhi_{BD} G = \bPhi_{AB} w_1^{-1} - \bPhi_{AD} w_2 = \bPhi_{BC} u_1 w_1^{-1} + \bPhi_{CD} u_2.
\end{equation*}
For our purposes it is convenient to attach the four letters $A,B,C,D$ to vertices of a simplex. Then edges of the simplex label points of the Veblen configuration, while its faces label lines of the configuration, see Figure~\ref{fig:Veblen-flip-tetr}. Notice that in such a representation the Veblen flip map $W^G$ has the same description as the previous map $W$. The only difference is that the variables $u$ and $w$ are attached to edges of the simplex, and the edge opposite to the arrow represents the point of the configuration with homogeneous coordinates on the left hand side of the corresponding linear equation with the coefficient equal to one. Such similarity allows to consider 
For our purposes it is important to have the description of the Veblen flip map in the pentagon property of the Veblen flip map map $W^G$. 
\begin{Prop}
The map $W^G:\DD^2\times \DD^2 \ni ((u_1,w_1), (u_2, w_2)) \dashrightarrow 
((u^\prime_1,w^\prime_1), (u^\prime_2 , w^\prime_2)) \in \DD^2\times \DD^2 $ given by equations \eqref{eq:WG} satisfies the functional pentagon equation
\begin{equation} \label{eq:W-pent-G}
W_{12}^V \circ W_{23}^U = W_{23}^Z \circ W_{13}^Y \circ W_{12}^X, \qquad \text{in} \qquad
\DD^2\times \DD^2 \times \DD^2,
\end{equation}
provided the parameters of the maps satisfy the relation
\begin{equation} \label{eq:UV-XYZ}
V=Yw_2, \qquad ZX =- Uw_2.
\end{equation}
\end{Prop}
\begin{proof}
By direct calculation
\begin{equation*}
W_{12}^V \circ W_{23}^U \left( \begin{matrix}
u_1 , w_1 \\ u_2 ,  w_2 \\ u_3 ,  w_3
\end{matrix} \right) = 
\left( \begin{array}{rl}
V w_1, & w_3 w_2 w_1 \\ 
- U w_2 w_1 u_1^{-1} , & V w_1 u_1^{-1}  \\ 
- u_3 w_2 u_2^{-1}, & U w_2 u_2^{-1}
\end{array} \right) 
\end{equation*}
and 
\begin{equation*}
W_{23}^Z \circ W_{13}^Y \circ W_{12}^X
\left( \begin{matrix}
x_1, y_1 \\ x_2, y_2 \\ x_3, y_3
\end{matrix} \right) =
\left( \begin{array}{rl}
Y w_2 w_1, & w_3 w_2 w_1 \\ 
Z X w_1 u_1^{-1} , & Y w_2 w_1 u_1^{-1}  \\ 
- u_3 w_2 u_2^{-1}, & - ZX u_2^{-1}
\end{array} \right) .
\end{equation*}
\end{proof}
\begin{figure}
\begin{center}
\includegraphics[width=12cm]{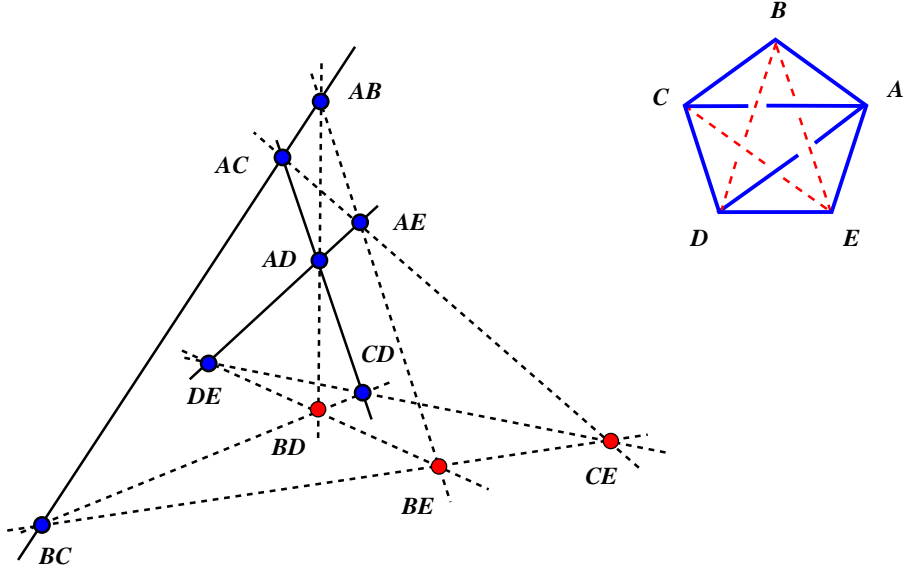}
\end{center}
\caption{Combinatorics of the Desargues configuration; the initial seven points of the configuration used to study the pentagonal property of the Veblen flip map are  distinguished by solid edges of the $4$-simplex, while the initial three lines of the configuration correspond to three distinguished triangular facets of the $4$-simplex.}
\label{fig:Desargues-simplex}
\end{figure}
To understand the geometric origin of the pentagon relation property of the Veblen flip let us start from seven points of the Desargues configuration which belong to three of its lines, see Figure~\ref{fig:Desargues-simplex}. By making a sequence of Veblen flips, which can be combinatorially represented in the same way like that for the normalization map $W$ (on Figure~\ref{fig:pentagon}), we can recover all the other points of the Desargues configuration. Conditions \eqref{eq:UV-XYZ} result from adjustment of homogeneous coordinates of two points $BE$ and $CE$ (the third new point $BD$ is not constructed in the "upper" way).

\subsubsection{Quantum reduction of the Veblen flip map}
Like in the case of the normalization map $W$ we will be interested in the ultra-local reduction of the Veblen flip map. In that case we also "fix" the gauge parameter $G$ of the map in the sense that it depends on its arguments as follows
\begin{equation} \label{eq:G-ab}
G(u_1,w_1,u_2,w_2) = (\alpha u_2 + \beta w_2 u_1) w_1^{-1}, \qquad \alpha, \beta \in \Bbbk,
\end{equation}
which effectively gives dependence of the map on two scalar (we assume both do not vanish simultaneously) parameters.  Then the map, which we denote from now on by $W^{(\alpha,\beta)}$ reads
\begin{equation} \label{eq:W-ab}
u^\prime_1 = \alpha u_2 + \beta u_1 w_2, \qquad w^\prime_1= w_1 w_2, \qquad
u^\prime_2 = -w_1 u_1^{-1} u_2 , \qquad w^\prime_2 = \alpha u_1^{-1} u_2 + \beta u_1 ,
\end{equation} 
and in the case of $\alpha=\beta=1$ reduces to formulas \eqref{eq:W}.

Again, the map $W^{(\alpha,\beta)}$ supplemented with the ultra-locality requirement and generic position assumptions selects the Weyl commutation relation, and in particular we have the following result (obtained in \cite{DoliwaSergeev-pentagon} for inverse of the map) which can be verified directly.
\begin{Cor} \label{cor:W-aut-ab}
The map 
$W^{(\alpha,\beta)}$ provides automorphism of  the division algebra $\Bbbk_q(u_1,w_1,u_2,w_2)$, and in consequence also gives Poisson automorphism of the field $\Bbbk(u_1,w_1,u_2,w_2)$.
\end{Cor}
The pentagonal property of the Veblen flip map $W^G$ transfers on the level of its ultra-local version $W^{(\alpha,\beta)}$ as follows.
\begin{Prop}
For $q$-commuting ultra-local arguments $u_i, w_i$, $i=1,2,3$, the map $W^{(\alpha,\beta)}$ given by \eqref{eq:W-ab} satisfies the functional pentagon equation
\begin{equation} \label{eq:W-pent-G-ab}
W_{12}^{(\alpha_V,\beta_V)}\circ W_{23}^{(\alpha_U,\beta_U)} = 
W_{23}^{(\alpha_Z,\beta_Z)} \circ W_{13}^{(\alpha_Y,\beta_Y)} \circ 
W_{12}^{(\alpha_X,\beta_X)}, 
\end{equation}
provided the (spectral) parameters of the maps satisfy the relation
\begin{equation} \label{eq:UV-XYZ-ab}
\alpha_X = \alpha_V \beta_Z, \qquad \alpha_Y = \alpha_U \alpha_V, \qquad 
\alpha_Z = \alpha_U \beta_X, \qquad
\beta_U = \beta_Y \beta_Z, \qquad \beta_V = \beta_X \beta_Y .
\end{equation} 
\end{Prop}
\begin{proof}
It is enough to check that the condition \eqref{eq:UV-XYZ} in the case of the gauge parameter $G$ given by \eqref{eq:G-ab} and the ultra-locality requirement reduces to the condition \eqref{eq:UV-XYZ-ab}.
\end{proof}
\begin{Rem}
The corresponding solution of the quantum pentagon equation was constructed in~\cite{DoliwaSergeev-pentagon}.
\end{Rem}

\section*{Acknowledgments}
I would like to thank Paolo M. Santini for the long standing joint work on integrable aspects of multidimensional quadrilateral lattices and their reductions, and Sergey M. Sergeev for  fruitful collaboration on the paper \cite{DoliwaSergeev-pentagon} and discussions.
The research was supported in part by the Ministry of Science and Higher Education grant No. N~N202~174739.

\bibliographystyle{amsplain}

\providecommand{\bysame}{\leavevmode\hbox to3em{\hrulefill}\thinspace}

\end{document}